%% file: 22-Journal-Neural-Lyapunov.tex
\documentclass[10pt]{IEEEtran}

 \usepackage{amsmath,amsthm,amstext,amsfonts,amssymb,mathrsfs} 

\usepackage{array}
\usepackage[caption=false,font=normalsize,labelfont=sf,textfont=sf]{subfig}
\usepackage{textcomp}
\usepackage{stfloats}
\usepackage{url}
\usepackage{verbatim}
\usepackage{graphicx}
\usepackage{cite} 
\usepackage{color}
\usepackage{algpseudocode}
\usepackage{algorithm}
\usepackage{algcompatible}
\usepackage{soul}

\DeclareMathOperator*{\argmin}{arg\,min}

\newcommand{\norm}[1]{\left\| #1 \right\|}%

\usepackage{amsthm}
\theoremstyle{definition}
\newtheorem{definition}{Definition}
\newtheorem{proposition}{Proposition}
\newtheorem{assumption}{Assumption}

\newtheorem{remark}{Remark}
\usepackage{makecell}

\usepackage{xcolor}

\begin{document}

\title{Distributed Learning of Neural Lyapunov Functions for Large-Scale Networked Dissipative Systems}


\author{Amit Jena, Tong Huang, S. Sivaranjani, Dileep Kalathil, and Le Xie
    	\thanks{Amit Jena, Dileep Kalathil and Le Xie are with the Department of Electrical and Computer Engineering at Texas A\& M University, College Station, TX, USA. Email:{\tt \{amit$\_$jena, dileep.kalathil, le.xie\}@tamu.edu }.
    	\par Tong Huang is with the MIT Laboratory of Information \& Decision Systems, Cambridge, MA, USA. Email:{\tt tongh@mit.edu}. 
    	\par  S. Sivaranjani is with the School of Industrial Engineering at Purdue University, West Lafayette, IN, USA. Email:{\tt sseetha@purdue.edu}. 
    	}
}

\maketitle

\begin{abstract}
This paper considers the problem of characterizing the stability region of a large-scale networked system comprised of dissipative nonlinear subsystems, in a distributed and computationally tractable way.  One standard approach to estimate the stability region of a general nonlinear system is to first find a Lyapunov function for the system and characterize its region of attraction as the stability region. However, classical approaches, such as sum-of-squares methods and quadratic approximation, for finding a Lyapunov function either do not scale to large systems or give very conservative estimates for the stability region. In this context, we propose a new distributed learning based approach by exploiting the  dissipativity structure of the subsystems.  Our approach has two parts: the first part is a distributed approach to learn the storage functions (similar to the Lyapunov functions) for all the subsystems, and the second part is a distributed optimization approach to find the Lyapunov function for the networked system using the learned storage functions of the subsystems. We demonstrate the superior performance of our proposed  approach  through extensive case studies in microgrid networks.
\end{abstract}

\begin{IEEEkeywords}
Lyapunov function, stability assessment, distributed learning, dissipativity 
\end{IEEEkeywords}

\input{01-Introduction}

\input{02-Modeling}

\input{03-Algorithm}

\input{05-CaseStudies}

\input{06-Conclusion}

\bibliographystyle{IEEEtran}
\bibliography{22-Journal-Neural-Lyapunov}

\begin{IEEEbiography}
{Amit Jena} received the Integrated M.S. degree in mathematics from National Institute of Technology, Rourkela, India, M.S. degree in electrical engineering from Iowa State University, IA, USA. He is currently a Ph.D. student in electrical engineering at Texas A\& M University, college station, TX, USA. His research interests includes reinforcement learning, meta learning, with application in power systems.  
\end{IEEEbiography}

\begin{IEEEbiography}
{Tong Huang} (Member, IEEE) received the Ph.D. degree in electrical engineering from Texas A\& M University, College Station, TX, USA, in 2021, where he is currently a Postdoctoral Researcher. From September 2018 to December 2018, he was a visiting Ph.D. student with the Laboratory for Information and Decision Systems, MIT. His industry experience includes an internship with ISO-New England in 2018 and with Mitsubishi Electric Research Laboratories in 2019. He received the Best Paper Award at the 2020 IEEE Power and Energy Society General Meeting, the Best Paper Award at the 54-th Hawaii International Conference on System Sciences, the Thomas W. Powell’62 and Powell Industries Inc., Fellowship, and the Texas A\& M Graduate Teaching Fellowship.
\end{IEEEbiography}

\begin{IEEEbiography}
{S. Sivaranjani} received the B.E. degree from PES Institute of Technology, Bengaluru, India, M.S. degree from the Indian Institute of Science, Bengaluru, India, and Ph.D. degree from the University of Notre Dame, IN, USA, respectively, all in electrical engineering. She is currently an Assistant Professor in the School of Industrial Engineering at Purdue University, West Lafayette, IN. Previously, she was Postdoctoral Researcher with the Department of Electrical and Computer Engineering, Texas A\& M University, College Station, TX, USA. Her research interests include data-driven and distributed control for large-scale networked systems, with applications to energy systems and transportation networks.
\end{IEEEbiography}

\begin{IEEEbiography}
{Dileep Kalathil} (Senior Member, IEEE) received his Ph.D. degree from the University of Southern California (USC) in 2014. From 2014 to 2017, he was a Postdoctoral Researcher with the Department of Electrical Engineering and Computer Sciences, University of California at Berkeley. He is currently an Assistant Professor with the Department of Electrical and Computer Engineering, Texas A\&M University. His main research focus is on reinforcement learning theory and algorithms, with applications in energy systems, communication networks and mobile robotics. He was a recipient of the Best Academic Performance from the EE Department, IIT Madras and the Best Ph.D. Dissertation Prize in the USC Department of Electrical Engineering, NSF CRII Award in 2019 and NSF CAREER award in 2021. 
\end{IEEEbiography}

\begin{IEEEbiography}
{Le Xie} (Fellow, IEEE) received the B.E. degree in electrical engineering from Tsinghua University, Beijing, China, in 2004, the M.S. degree in engineering sciences from Harvard University, Cambridge, MA, USA, in 2005, and the Ph.D. degree from the Department of Electrical and Computer Engineering, Carnegie Mellon University, Pittsburgh, PA, USA, in 2009. He is currently a Professor with the Department of Electrical and Computer Engineering, Texas A\& M University, College Station, TX, USA. His research interests include modeling and control of large-scale complex systems, smart grids application with renewable energy resources, and electricity markets.
\end{IEEEbiography}

\vfill

\end{document}

%% file: 01-Introduction.tex
\section{Introduction}

Large-scale networked systems formed by interconnection of a several nonlinear subsystems are encountered in several practical applications such as infrastructure networks including power networks, transportation networks, and communications networks. For example, in power networks, a collection of renewable generators, storage, and loads in a small area can be modeled as a single microgrid with nonlinear dynamics, and a number of such microgrids form a  network which balances the supply and demand of electric power in a larger geographical area. Analyzing the stability margins of the dynamics of such infrastructure systems is crucial in ensuring their safe and reliable operation. One standard approach to estimate the stability region of a general nonlinear system is to first find a Lyapunov function for the system and characterize its region of attraction as the stability region. The sum-of-squares approach is one popular method for finding a Lyapunov function for a dynamical system \cite{parrilo2000structured, henrion2005positive,jarvis2003some,topcu2009robust, topcu2008local}. However, sum-of-squares approaches typically do not scale well to large systems, since a large number of semidefinite programs need to be solved for the sum-of-squares decomposition of polynomial systems even with a few states \cite{parrilo2000structured}.  Another approach is to employ local linearizations and use quadratic  approximations to find Lyapunov functions. However, this approach is typically conservative in the sense that stability can only be certified in a small vicinity of an equilibrium point of a nonlinear system, which may be insufficient to cover the normal range of operation in several practical applications such as microgrids in power systems \cite{huang2021neural}.

Due to the large size and continually expanding scale of many real-world networked systems, it is especially important to develop distributed approaches to assess the stability of such systems. Analyzing the stability of a general networked system in a distributed manner from the stability assessment of its subsystems is a complex open problem.  Such a distributed stability assessment, however, is possible for a class of networked systems where each subsystem satisfies passivity or dissipativity properties  \cite{vidyasagar1979new,moylan1979tests, arcak2016networks, agarwal2020distributed}. However, these results are typically confined to local stability assessments based on linearized system dynamics. Thus,  \textit{how to find a  Lyapunov function (and the corresponding stability region) for a large scale networked system in a distributed and computationally tractable way} remains an important open problem.

In this paper, we consider the problem of finding a  Lyapunov function and estimating its associated region of attraction in a distributed manner for a large scale networked system comprised of a number of subsystems. We assume that each subsystem is nonlinear and satisfies the dissipativity property \cite{arcak2016networks}.  Our approach has two stages: in the first stage, we develop a distributed learning approach to learn the storage functions (similar to the Lyapunov functions) for all the subsystems; in the second stage,  we develop a distributed optimization approach to find the Lyapunov function for the networked system using the learned storage functions of the subsystems. Our main algorithm essentially is an iteration over these two parts until some convergence criterion is satisfied.

Specifically, in the first stage,  we use a neural network approximation to learn a storage function for each subsystem such that the subsystem satisfies the dissipativity property. Our approach is inspired from the recent works which use the data-based  function approximation capability of neural networks to \textit{learn} a Lyapunov function for nonlinear systems  \cite{kolter2019learning, chang2019neural, huang2021neural,huang2021HICSS} in a centralized way. In addition to this empirical learning, we also use a satisfiability modulo theories (SMT) solver based falsifier that finds counterexamples to verify this local dissipativity property using the storage function; when no counterexample is found by the falsifier, the subsystem  is provably dissipative \cite{chang2019neural}. In the second stage, we use the fact that the stability of the networked system can be guaranteed by verifying some conditions on the  dissipativity of the subsystems and the coupling between dissipative subsystems \cite{arcak2016networks}. For efficient implementation, we formulate and solve this as a distributed optimization problem using an alternating direction method of multipliers (ADMM) procedure which iteratively updates the storage function for each subsystem until the network-level stability conditions are satisfied. This step also allows us to compute the Lyapunov function for the networked system and estimate its region of attraction. Additionally, we learn local controllers at the subsystem-level that enhance the stability region of the networked system.

We demonstrate the performance of our proposed distributed learning approach for stability assessment through extensive case studies in microgrid networks. Our choice of microgrid networks is motivated by their importance in real-world power systems  and their natural structure as a network of smaller nonlinear subsystems \cite{olivares2014trends, katiraei2008microgrids, zamora2016multi}. In recent years, the stability assessment problem for a network of microgrids has also attracted a lot of attention from the power systems and controls community \cite{shamsi2014stability,he2019small,song2017distributed}.

A preliminary version of this work was published as a conference paper  \cite{jena2021distributed}. In this work, we significantly expand and extend the results in \cite{jena2021distributed}. In particular, we learn a local controller for each subsystem that improves the stability of the network system and makes it robust to bounded disturbances. We also consider a larger network based case study which clearly illustrates the superior performance our distributed learning approach both in terms of finding a larger stability region and in reduced training/computation time.  In summary, the key contributions of this paper are as follows:
\begin{enumerate}
    \item We develop a distributed approach to learn the Lyapunov function and the stability region of a large scale networked system  by exploiting the dissipativity property of its subsystems and leveraging recent advances in neural network based nonlinear function approximation/learning using data.
    \item We learn a local controller for each subsystem that further improves the stability of the entire networked system.
\end{enumerate}

This paper is organized as follows.  Section \ref{sec:formulation} gives the preliminaries and basic problem formulation,  Section \ref{sec:algorithm} describes our distributed learning based algorithm,  Section \ref{sec:case_studies} presents the case studies,  and Section \ref{sec:conclusion} concludes with future research directions.

%% file: 02-Modeling.tex
\section{Problem Formulation}
\label{sec:formulation}

We consider a large scale  nonlinear networked system formed by the interconnection of  $n$ nonlinear  subsystems. The dynamics of the $i$th subsystem is given by 
\begin{align}
\label{eq:subsystem}
        \dot{x}_i = f_i(x_i, u_i, v_i), \quad
        y_i = h_i(x_i,u_i, v_i),
\end{align}
where $x_i \in \mathbb{R}^{d^{x}_{i}}$, $u_i \in \mathbb{R}^{d^{u}_{i}}$, $v_i \in \mathbb{R}^{d^{v}_{i}}$, and $y_i\in \mathbb{R}^{d^{y}_{i}}$ are the state, primary  input,  secondary (local) input, and output of the subsystem $i$, respectively. The system dynamics and system output are specified by the continuously differentiable functions ${f}_i$ and ${h}_i$. The  state, primary input, secondary input and output of the networked system are then defined as $x=[x_1^{\top}, \ldots, x_n^{\top}]^{\top}$, $u=[u_1^{\top}, u_2^{\top}, \ldots, u_n^{\top}]^{\top}$,  $v=[v_1^{\top}, v_2^{\top}, \ldots, v_n^{\top}]^{\top}$, and $y=[y_1^{\top}, y_2^{\top}, \ldots, y_n^{\top}]^{\top}$ respectively. The overall dynamics of the networked system is specified as 
\begin{align}
    \label{eq:dynamic_whole_1}
    \dot{x} = f(x, u, v), \quad 
    y = h(x, u, v),
\end{align}
where $f=[f_1^{\top}, \ldots, f_n^{\top}]^{\top}$, and $h=[h_1^{\top}, \ldots, h_n^{\top}]^{\top}$. 

We assume that the subsystems are coupled through a primary and secondary input as
\begin{equation} 
\label{eq:couple_eq}
    u = g^{\textrm{pri}}(y), \quad  v_{i} = g^{\textrm{sec}}_{i}(y_{i}),
\end{equation}
where $g^{\textrm{pri}}$ and $g^{\textrm{sec}}_{i}$ specify the  primary  control law and secondary (local) control law, respectively. We adopt this two-level control structure for the following reason. Many large scale real-world networked systems, such as power systems, use  a hierarchical control architecture with a dedicated  secondary controller at each subsystem that maps its local output to a local control input \cite{sivaranjani2020distributed, okou2005power}. Such secondary controllers are known to critically improve the stability of the networked system using only local information \cite{okou2005power}. Our separate modeling on the primary ($ u_{i}$)  and secondary ($ v_{i}$) control inputs is motivated from such real-world systems. Later, we will also show experimental results which demonstrate the role of the secondary controller in improving the stability region of the networked system.

Without loss of generality, assume that the equilibrium point of \eqref{eq:dynamic_whole_1}-\eqref{eq:couple_eq} is given by $(x^{*}, u^{*}, v^{*}, y^{*}) = o$, where $o$ is the origin. Also,  denote  $(x^{*}_{i}, u^{*}_{i}, v^{*}_{i}, y^{*}_{i})$ as $o_{i}$. For simplicity, we further restrict our consideration to a linear approximation of the coupling \eqref{eq:couple_eq} between subsystems, given by
\begin{equation}
\label{eq:linear-couple_eq}
    u = M y, \quad v_{i} = K_{i} y_{i},
\end{equation}
where $M$ and $K_{i}$ are the Jacobian matrices of $g^{\textrm{pri}}$ and $g^{\textrm{sec}}_{i}$ evaluated at $o$ and $o_{i}$, respectively. The overall dynamics of the networked system constituted by the interconnection of these $n$ subsystems is now completely specified by \eqref{eq:dynamic_whole_1} and \eqref{eq:linear-couple_eq} for all $(x, u, v, y) \in B_{o}$, where $B_{o}$ is a neighborhood around the origin $o$.

We now define the \textit{dissipativity}  property for the subsystems.
\begin{definition}[Dissipativity \cite{arcak2016networks}]
\label{defn-dissiativity}
    The subsystem $i$ described by  \eqref{eq:subsystem} is said to be \textit{dissipative} with respect to the \textit{supply rate function} $r_i(u_i, v_i,  y_i)$ if there exists  a continuously differentiable \textit{storage function} $\tilde{V}_i:\mathbb{R}^{d^{x}_{i}}\rightarrow \mathbb{R}$ such that  $\tilde{V}_i(o_{i}) = 0$, $\tilde{V}_i(x_i) \geq 0$, and 
    \begin{align}
    \label{eq:defn-dissiativity}
         \dot{\tilde{V}}_i(x_i) = \nabla \tilde{V}_i(x_i)^\top f_i(x_i, u_i, v_i) \leq r_{i}(u_i, v_i, y_i),
    \end{align}
for all $(x_i, u_i, v_i, y_i) \in  B_{o_i}$. 
\end{definition}

In this work, we assume that each subsystem satisfies the {dissipativity} property with respect to a quadratic supply rate function, as stated below. 
\begin{assumption}[Dissipative subsystem]
\label{assum:dissiativity}
Every subsystem $i\in\{1,2,\ldots,n\}$, described by \eqref{eq:subsystem}, is \textit{{dissipative}} with respect to the quadratic supply rate function given by
\begin{align}
\label{eq:supply-rate-ri}
    r_i(u_i, v_i, y_i) = \left[\begin{array}{c}
u_{i} \\
y_{i} \\
v_{i} \\
y_{i}
\end{array}\right]^{\top}
\left[\begin{array}{cc}
R_{i} & 0 \\
0 & R_i
\end{array}\right]
\left[\begin{array}{c}
u_{i}\\
y_{i}\\
v_{i} \\
y_{i}
\end{array}\right],
\end{align}
where $R_{i}=\left[\begin{array}{ll}R_{i}^{11} & R_{i}^{12} \\ R_{i}^{21} & R_{i}^{22}\end{array}\right]$. 
\end{assumption}

Next, we  state the  definitions of Lyapunov function and region of attraction (RoA). Note that we will  use the terms RoA and stability region interchangeably  in the following.
\begin{definition}[Lyapunov function and region of attraction \cite{khalil2002nonlinear}]
\label{def:lyapunov}
 For the system \eqref{eq:dynamic_whole_1}, a continuous differentiable scalar function $V: \mathbb{R}^{d^{x}} \rightarrow \mathbb{R}, d^{x} = \sum^{n}_{i=1} d^{x_{i}}$,  is a \textit{strict Lyapunov function} valid in a region $\mathcal{V}_{\bar{c}}:=\{x:\norm{x}_2\le \bar{c}\}$ if the following conditions are satisfied: \\
 $(i)$ $V(o) = 0$, $(ii)$  $V(x)>0$ for all $x\in \mathcal{V}_{\bar{c}} \backslash o$, and $(iii)$  $\dot{V} = \nabla V(x) \dot{x}<0$, for all $x\in \mathcal{V}_{\bar{c}} \backslash o$. \\ Further, the region $\mathcal{V}_{\bar{c}}$ is defined as the  \textit{region of attraction} (RoA) of the equilibrium point $o$, that is, $ x(0)\in \mathcal{V}_{\bar{c}} \implies \lim_{t \rightarrow \infty} x(t)=o.$
\end{definition}

The stability of the networked system  described by \eqref{eq:dynamic_whole_1} and \eqref{eq:linear-couple_eq} can be assessed by identifying a Lyapunov function  satisfying Definition \ref{def:lyapunov} and by characterizing the RoA.  In general, sum-of-squares approaches \cite{parrilo2000structured,henrion2005positive,jarvis2003some,topcu2009robust,topcu2008local} can be used to assess the stability of  nonlinear systems. Alternatively, local linearizations of  the dynamics  may be used to compute a quadratic Lyapunov function \cite{chiang1989study}. However, these approaches often fail to find a meaningful Lyapunov function and RoA for large-scale, high dimensional,  networked systems because of the following challenges. Firstly, sum-of-squares approaches do not scale well computationally and will quickly become intractable when the system dimension increases \cite{parrilo2000structured,henrion2005positive,jarvis2003some,topcu2009robust,topcu2008local}. Secondly, sum-of-squares approaches and quadratic approximation based approaches are typically very conservative in their estimation of RoA even for small systems \cite{chang2019neural, huang2021neural}. This may lead to the design of conservative controllers and sub-optimal system operation. 

In order to overcome these  challenges, some recent works have exploited the data-based  function approximation capability of neural network to \textit{learn} a Lyapunov function for nonlinear systems  \cite{kolter2019learning, chang2019neural, huang2021neural}, with  impressive empirical performance in estimating the stability region for nontrivial nonlinear systems. However, the training time, number of samples needed, and computational complexity of these approaches  also increase exponentially in the system dimension, which can be prohibitively high for  a networked system constituted by a large number of smaller subsystems. In this work, we propose a distributed learning approach that exploits the dissipativity property of the subsystems to overcome these challenges and efficiently learn the Lyapunov function and stability region of the large-scale networked system.

%% file: 03-Algorithm.tex
\section{Distributed Learning of Neural Lyapunov Function for the Networked System}\label{sec:nlf}
\label{sec:algorithm}

In this section, we  propose a two-stage distributed approach to  learn a Lyapunov function for the networked system and characterize its RoA.

\subsection{Lyapunov Function for the Networked System from the Storage Functions of the Subsystems}
\label{sec:subsec-decomposition-result}
Our key idea is to reduce the  dimension of the problem by  leveraging the fact that the Lyapunov function for the networked system can be constructed from the  storage functions of the subsystems   under the  dissipativity assumption (Assumption \ref{assum:dissiativity}) \cite{arcak2016networks}. We  formally state this result below.
\begin{proposition}
\label{prop:diss_to_stab}
Consider the networked system  specified by \eqref{eq:dynamic_whole_1} and \eqref{eq:linear-couple_eq}. Let Assumption \ref{assum:dissiativity} hold, and $r_{i}$  as specified in \eqref{eq:supply-rate-ri}   be the supply rate function and $\tilde{V}_i$ be the corresponding storage function of subsystem $i$. If 
\begin{eqnarray}
\label{eq:global_diss_cond}
&\left[\begin{array}{c}
{M} \\
{I} \\
{K} \\
{I}
\end{array}\right]^{\top}
\left[\begin{array}{cc}
R & 0 \\
0& R
\end{array}\right]
\left[\begin{array}{c}
{M}\\
{I}\\
{K} \\
{I}
\end{array}\right] \preceq 0, 
\end{eqnarray}
where 
\begin{align*}
  R &=\nonumber\left[\begin{array}{cccccc} R_{1}^{11} & & &  R_{1}^{12} & & \\ & \ddots & & & \ddots & \\ & &  R_{N}^{11} & &  &  R_{N}^{12} \\  R_{1}^{21} & & &  R_{1}^{22} & & \\ & \ddots & &  & \ddots & \\ & &  R_{N}^{21} & &  &  R_{N}^{22}\end{array}\right],
\end{align*}
$K = \text{diag}([K_1, \hdots, K_n])$ is a block diagonal matrix, and $I$ is an identity matrix of appropriate order, then,  \\
$(i)$ the networked system has a stable equilibrium at $o$, and\\
$(ii)$ $V(x)=\sum\limits_{i=1}^n \tilde{V}_i(x_i)$ is a  Lyapunov function for the networked system.
\end{proposition}

\begin{proof}
The proof follows a direct extension of the result given in \cite[Chapter 2]{arcak2016networks} as follows. Consider the networked system specified by \eqref{eq:subsystem}-\eqref{eq:linear-couple_eq}.  
Let $\tilde{V}_i$ be a valid storage function for subsystem $i$ with respect to the supply rate function $r_{i}$  such that it satisfies the dissipativity condition  \eqref{eq:defn-dissiativity}, for all $i \in \{1, \ldots, n\}$. Now, considering the function $V(x)=\sum\limits_{i=1}^n\tilde V_i(x_i)$, we have
\begin{equation}
\label{eq:appendix:diss_def_1}
        \dot{V}(x) = \sum\limits_{i=1}^{N} \nabla \tilde{V}_i  f_i(x_i, u_i, v_i) <     \sum\limits_{i=1}^{n} r_i(x_i, u_i, v_i), 
\end{equation}
where the inequality follows from the dissipativity condition \eqref{eq:defn-dissiativity}. Assuming the specific form for $r_{i}$ in \eqref{eq:supply-rate-ri}, we  have
\begin{align}
      \sum\limits_{i=1}^{n} r_i(x_i, & u_i, v_i) = \sum\limits_{i=1}^{N} \left[\begin{array}{c}
u_{i} \\
y_{i} \\
v_{i} \\
y_{i}
\end{array}\right]^{\top}
\left[\begin{array}{cc}
R_{i} & 0 \\
0 & R_i
\end{array}\right]
\left[\begin{array}{c}
u_{i} \\
y_{i} \\
v_{i} \\
y_{i}
\end{array}\right] \nonumber \\
\label{eq:appendix:diss_def_2}
&=y^\top
\left[\begin{array}{c}
{M} \\
{I} \\
{K} \\
{I}
\end{array}\right]^{\top}
\left[\begin{array}{cc}
R & 0 \\
0 & R
\end{array}\right]
\left[\begin{array}{c}
{M}\\
{I}\\
{K} \\
{I}
\end{array}\right]
y,
\end{align}
where the last equality follows from the definition of $u$ and $v$ with $R$ as specified in the proposition statement. 

Now, if \eqref{eq:global_diss_cond} is satisfied, using \eqref{eq:appendix:diss_def_1}-\eqref{eq:appendix:diss_def_2},  it immediately implies that $\dot{V}(x) < 0$. The positivity of $V$ also directly follows from that of $\tilde{V}_{i}$s. Thus, $V$ satisfies the condition to be a valid Lyapunov function for the networked system. 
\end{proof}

\begin{remark}
Proposition \ref{prop:diss_to_stab} will enable us to reduce the high dimensional problem of finding a Lyapunov function for the networked system to $n$ smaller problems of finding the  storage functions of the subsystems. For example, assuming $d^{x}_{i} = d$ for all $i$, instead of an $nd$ dimensional problem of finding the  Lyapunov function for the networked system, we only need to solve $n$ different $d$ dimensional problems of finding the  storage functions of the subsystems. Since the computational complexity typically grows exponentially in the dimension of the system, this will reduce the exponential dependence on $n$ to a linear dependence on $n$. This reduction in complexity is especially significant for large-scale systems formed by the interconnection of a large number of small subsystems.  
\end{remark}

\subsection{Learning  Neural Storage Functions and Secondary Controllers for Subsystems}
\label{sec:subsec-learning-storage-function}

Neural networks can efficiently learn/approximate a broad class of nonlinear functions using data,  leading to their success in many engineering applications. Recent works have exploited this property of neural networks to learn  Lyapunov functions for nonlinear systems \cite{kolter2019learning, chang2019neural, huang2021neural}. We follow a similar approach for learning the  storage function of individual subsystems. In addition to learning the  storage function,  we will also learn the  secondary controller $K_{i}$, in order to maximize the stability region of the networked system. We achieve these goals by designing a loss function for the neural network learning that  penalizes  any violation of the dissipativity condition given in \eqref{eq:defn-dissiativity}, and finding the parameter which minimizes this loss. 

For each susbsyetm $i$, we create the  training datasets $\mathcal{D}^{x}_i = \{x_i^1, \hdots, x_i^N\}$ and $\mathcal{D}^{u}_i = \{u_i^1, \hdots, u_i^N\}$ by drawing $N$ samples from their domains according to some distribution. We represent the  storage function $\tilde{V}_{i}$ as $\tilde{V}_{\theta_{i}}$ where  $\theta_{i}$ is the parameter of the corresponding neural network. We then define the empirical loss function for subsystem $i$ as
\begin{align}
     \label{eq:diss_loss_def}
    &L(\theta_{i}, K_{i}) =  \tilde{V}_{\theta_{i}}^{2}(0) + \frac{1}{N} \sum_{j=1}^{N}\Big[\max (-\tilde{V}_{\theta_{i}}(x_{i}^j), 0) \nonumber \\
    &+\max \big(0, \nabla \tilde{V}_{\theta_{i}}(x_{i}^j)^{\top} {f}_{i}(x_{i}^j, u_{i}^j, K_{i} y^{j}_{i}\big) -r_{i}\big(u_{i}^j, K_{i} y^{j}_{i}, y_{i}^j)\big)\Big],
\end{align}
where the  secondary control input $v^{j}_{i}$ is specified as $ K_{i} y^{j}_{i}$. We will train a neural network for each subsystem using the above loss function until convergence.

A neural network based  learning algorithm can  give only an \textit{empirical guarantee} with respect to its loss function and training data. So,  the  parameters of the neural network obtained by minimizing the above empirical loss function may not \textit{provably} yield a storage function that satisfies the condition for  dissipativity \eqref{eq:defn-dissiativity}. We overcome this issue as described below.

Satisfiability Modulo Theories (SMT) solvers  are widely used to verify correctness of symbolic arithmetic expressions in a specified region $\mathcal{B}$ in the space of real numbers \cite{barrett2018satisfiability}. In particular, given a symbolic expression (also known as logic formula) $\Phi(x)$, an SMT solver uses a falsification technique to identify \textit{counterexamples} $x \in \mathcal{B}$ where $\Phi(x)$ is violated. When no such counterexamples exists, $\Phi(x)$ is said to be provably valid in $\mathcal{B}$. Recent works have implemented SMT solvers to validate Lyapunov functions \cite{huang2021neural, chang2019neural}. We follow a similar procedure to validate the storage functions learned by our neural network. Specifically, by using an SMT solver at every neural network training step, we validate $\tilde{V}_{\theta_i}$ by checking the following first order logic formula over real numbers:
\begin{align}
    \label{eq:SMT_diss_logic}
    \Phi(x_i,u_i) =& 
    \Big( \varepsilon_{u_i}\geq\|x_{i}\|^{2} \geq \varepsilon_{l_i}\Big) \wedge \Big(\varepsilon_{u_i}\geq \|u_{i}\|^{2} \geq \varepsilon_{l_i}\Big) \wedge \nonumber\\
    &\Big(\tilde{V}_{\theta_{i}}(x_i) \leq 0 ~\vee~ \nabla V_{\theta_{i}}(x_i)^{\top} {f}_{i}(x_i, u_i,K_{i} y_{i})\nonumber\\
    &-r_{i}(u_i, K_{i} y_{i}, y_i))\geq 0\Big),
\end{align}
where $\epsilon_{l_i}$ is a small lower-bound to avoid arithmetic underflow, and $\epsilon_{u_i}$ is an upper-bound that defines the verification region $\mathcal{B}_i$. We use a specific SMT solver called dReal \cite{gao2013dreal} to verify the above logic formula. 

Under the conditions   as described in \cite{chang2019neural, gao2013dreal, gao2012delta}, an SMT solver always returns a counterexample if there exists any. We add these counterexamples to our training dataset $\mathcal{D}_i$, and retrain the neural network to improve the validity of $\tilde{V}_{\theta_i}$ over $\mathcal{B}_i$. When no counterexample is returned, we complete the neural network training, and declare $\tilde{V}_{\theta_i}$ as a valid storage function. Performing this SMT verification based neural network training for all subsystems, we get $(\tilde{V}_{\theta_i}, K_{i})$ which can provably  certify  the dissipativity of the subsystem $i$ with respect to the supply rate function $r_{i}$, for all $i\in\{1, \hdots, n\}$.


\subsection{Distributed Learning of  Neural Lyapunov Function for the Networked System}

After learning the  neural storage functions for all subsystems, we can form a candidate neural Lyapunov function $V$ for the networked system as, $V(x) = \sum_{i} \tilde{V}_{\theta_{i}}(x_i)$. However, the fact that $\tilde{V}_{\theta_{i}}$ is a valid storage function for subsystem $i$ with respect to the supply rate function $r_{i}$, for all $i$, does not immediately imply that $V$ is a valid Lyapunov function for the networked system. The linear matrix inequality (LMI) \eqref{eq:global_diss_cond}  involving $M, K$ and $R_{i}$s (which represents $r_{i}$s) in Proposition \ref{prop:diss_to_stab} gives a sufficient condition for such a $V$ to be valid Lyapunov function for the networked system. Therefore,  if  \eqref{eq:global_diss_cond}  is not satisfied, we change $R_{i}$s (equivalently $r_{i}$s) and relearn the storage functions for the subsystems with respect to these new supply rate functions. We repeat this procedure until convergence. We employ  the alternating direction method of multipliers (ADMM) approach for  the efficient implementation and guaranteed convergence of this proposed procedure \cite{arcak2016networks}, as described below.

To bring the condition \eqref{eq:global_diss_cond} into the canonical optimization form, we will define the indicator function $\mathbb{I}_{\text{global}}$ as
\begin{align}
    & \mathbb{I}_{\text{global}}(R_1,\hdots,R_n) :=
\begin{cases}
& 0 \qquad \text{if \eqref{eq:global_diss_cond} holds true}\\
& \infty \qquad \text{otherwise.}
\end{cases}
\end{align}
Similarly, to bring the dissipativity condition \eqref{eq:defn-dissiativity} of each subsystem into the optimization form, we define $n$ indicator functions $\mathbb{I}_{\text{Local},i}$ as
\begin{align}
    & \mathbb{I}_{\text{local},i}(R_i,\tilde V_{\theta_i}) := 
\begin{cases}
& 0  \qquad \text{if \eqref{eq:defn-dissiativity}  holds true}\\
& \infty \qquad \text{otherwise.}
\end{cases}
\end{align}

The problem of finding a Lyapunov function for the networked system can now be written as the following optimization problem:
\begin{align}
    \label{eq:ADMM_opt}
    \min_{(R_{i}, \tilde V_{\theta_i}, Z_{i})^{n}_{i=1}} ~~~&\sum_{i=1}^{n} \mathbb{I}_{\text {local}, i}\left(R_{i}, \tilde V_{\theta_i}\right) +\mathbb{I}_{\text {global}}\left(Z_{1}, \ldots, Z_{n}\right) \nonumber\\
\text{s.t.} ~~~& R_{i}-Z_{i}=0 \quad \text { for } i=1, \ldots, n.
\end{align}
We construct the optimization problem in above way because through the auxiliary variables $Z_{i}$s, this network-level optimization problem can  be divided into  $n+1$ smaller sub-problems and can be solved in a distributed manner.  The iterative steps of the algorithm are as follows:
\begin{subequations}
\label{eq:RZS-update}
\begin{align}
\label{eq:R-update}
     R_i^{k+1} &=\argmin_{R_{i}~\text{s.t.}~ \eqref{eq:defn-dissiativity}} \|R_i-Z_{i}^{k}+S_{i}^{k}\|_{F}^{2}, \\
     \label{eq:Z-update}
(Z^{k+1}_{i})^{n}_{i=1} &= \argmin_{(Z_{i})^{n}_{i=1}~\text{s.t.}~ \eqref{eq:global_diss_cond}} \sum_{i=1}^{n}\|R_{i}^{k+1}-Z_{i}+S_{i}^{k}\|_{F}^{2}, \\
\label{eq:S-update}
 S_i^{k+1} &= R_{i}^{k+1}-Z_{i}^{k+1}+S_{i}^{k}, 
\end{align}
\end{subequations}
where $(S_i)_{i=1}^n$ are matrices to have consensus between $(R_i)_{i=1}^n$ and $(Z_i)_{i=1}^n$, and $\|.\|_{F}$ is the Frobenius norm.

The first step \eqref{eq:R-update} finds a new supply rate $R^{k+1}_{i}$ close to the auxiliary variable $Z^{k}_{i}$ given the consensus variable  $S^{k}_{i}$. We then obtain the storage function  $\tilde{V}_{\theta_i}$ and secondary controller $K_{i}$ with respect to $R^{k+1}_{i}$ through the neural network training described in Section \ref{sec:subsec-learning-storage-function}. We emphasize that this step can be implemented in parallel for the subsystems. Since parallel neural network training can be performed at scale using the modern GPU architecture, this approach offer significant reduction in training time.  The second step \eqref{eq:Z-update} is to impose the Lyapunov condition for the networked system. It finds (auxiliary) supply rate functions $Z^{k+1}_{i}$ close to  $R^{k+1}_{i}$ given the consensus variables $S^{k}_{i}$ such that they satisfy the sufficient condition \eqref{eq:global_diss_cond} to get a Lyapunov condition for the networked system. While this update step is centralized, we note that this is a simple convex optimization problem and does not involve any neural network training. Finally, step \eqref{eq:S-update} updates the consensus variables and the iteration continues until convergence.

We summarize the stages in the above subsections as a concise algorithm in Algorithm \ref{algor:DNL-Algo}.  

\begin{algorithm}
    \begin{algorithmic}[1]
\caption{ Distributed Neural Lyapunov Function Learning Algorithm}
\label{algor:DNL-Algo}
    \STATE {\textbf{Input:}} Initial supply rate matrices $(R_{i})^{n}_{i=1}$,  initial consensus metrics $(S_{i})^{n}_{i=1}$, input-output coupling matrix $M$, tolerance level $\eta$.
   \REPEAT
   \STATE Learn neural storage functions $(\tilde{V}_{\theta_{i}})^{n}_{i=1}$ and secondary controllers $(K_{i})^{n}_{i=1}$ according to the procedure described in Section \ref{sec:subsec-learning-storage-function} with respect to the current supply rate functions $(R_{i})^{n}_{i=1}$
  \STATE Update auxiliary variables $(Z_{i})^{n}_{i=1}$ according to \eqref{eq:Z-update}
  \STATE Update consensus variables $(S_{i})^{n}_{i=1}$ according to \eqref{eq:S-update}
  \STATE Update supply rate functions $(R_{i})^{n}_{i=1}$ according to \eqref{eq:R-update}
   \UNTIL{($\sum_i||R_i-Z_i||>\eta$)}
       \STATE {\textbf{Output:}}   Lyapunov function $V = \sum^{n}_{i} \tilde{V}_{\theta_{i}}$ for the networked system
\end{algorithmic}
\end{algorithm}

%% file: 05-CaseStudies.tex
\section{Case Studies}
\label{sec:case_studies}

In this section, we demonstrate the performance of our algorithm by evaluating it on three test cases: $(i)$ Three microgrids network with conventional angle droop control, $(ii)$  Four microgrids network with quadratic voltage droop control, and $(iii)$  IEEE 123-node test feeder network with conventional voltage and angle droop controls. We compare the  performance of our algorithm with the following approaches for finding the Lyapunov functions. 
\begin{itemize}
    \item\textit{Quadratic control Lyapunov  (QCL) function},    which  computes the Lyapunov function by linearizing the  dynamics around the origin  \cite{chang1995direct}. Additionally, a linear quadratic regulator (LQR) based secondary controller is designed for the improved stability of the system.
    \item \textit{Quadratic Lyapunov (QL) function}, which finds the Lyapunov function in the same way as QCL  but without a secondary controller. 
    \item \textit{Centralized neural control Lyapunov (CNCL) function}  finds a Lyapunov function through neural network based centralized learning  \cite{chang2019neural}. It also learns a secondary controller through neural network based learning. 
    \item \textit{Centralized neural Lyapunov function} (CNL) is the same as CNCL  but without a secondary controller  \cite{huang2021neural}. 
\end{itemize}

We evaluate two version of our algorithms:
\begin{itemize}
    \item \textit{Distributed neural control (DNCL) Lyapunov function}  learns the the Lyapunov function as described in Algoirthm \ref{algor:DNL-Algo}. 
    \item \textit{Distributed neural  Lyapunov (DNL) function} is same as DNCL  but without a secondary controller.
\end{itemize}

\subsection{Dynamics of the Network of Microgrids }
Since our case studies focuses on the network of microgirds, we first give a brief description about their dynamics and the classical control schemes used in the literature. 

\begin{figure}[h]
    \centering
    \includegraphics[width = \linewidth]{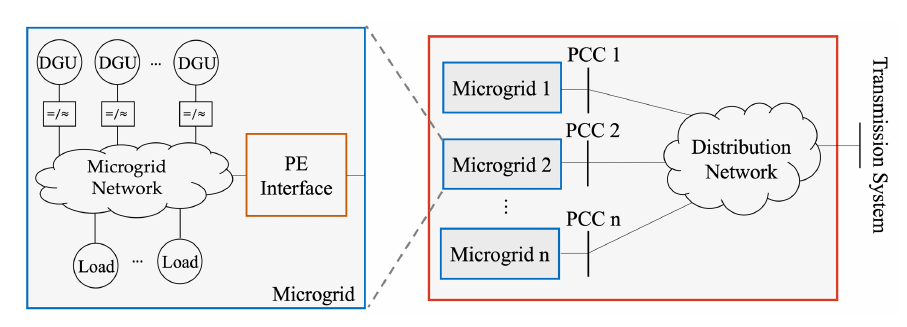}
    \caption{A distribution system consisting of networked microgrids  \cite{huang2021neural}. The internal structure of a single microgrid is shown in the left box.}
    \label{fig:networked_MG}
\end{figure} 

Fig. \ref{fig:networked_MG} shows a network of microgrids  where each microgrid comprises of several distributed generation units (DGUs), loads and power electronics (PE) interfaces \cite{zamora2016multi}. We assume that the microgrids function in a grid-connected mode and they  are networked with each other via their points of common coupling (PCCs) and distribution lines. The power injections to the PCCs are governed by the standard AC power flow equations as follows:
\begin{subequations}
\begin{align}
    P_i &= \sum_{k\ne i}E_iE_kY_{ik} \cos(\delta_{ik} - \gamma_{ik}) + E_i^2G_{ii}\label{eq:P_def},\\
     Q_i &= \sum_{k\ne i}E_iE_kY_{ik} \sin(\delta_{ik} - \gamma_{ik}) - E_i^2B_{ii}
     \label{eq:Q_def},
\end{align}
\end{subequations}
where $\delta_{ik} = \delta_i -\delta_k$, $E_i$ and $\delta_i$ are voltage magnitude and phase angle of the $i$\text{th} microgrid;  $P_i$ and $Q_i$ are active and reactive power injections at the $i$\text{th}  PCC; $G_{ii}+\emph{j} B_{ii}$ is the $i$\text{th} diagonal element in the admittance matrix; and $Y_{ik}$ and $\gamma_{ik}$ are the magnitude and angle of the  $(i,k)$\text{th} element in the admittance matrix respectively.

We consider two different control schemes that are widely used in  the literature \cite{guerrero2010hierarchical,simpson2016voltage}.

\subsubsection{Conventional droop control}
\label{sec:conv_droop}
In general, each microgrid employs a primary controller to stabilize the network and proportionally share the load among all the DGUs. This is achieved by setting up a control loop at a PE interface placed next to the microgrid. Droop controllers are the most extensively used control schemes, and they are based on active or reactive power decoupling \cite{guerrero2010hierarchical}. The droop control based  dynamics at $i$\text{th} microgrid is given by
\begin{subequations}
\label{eq:droop_dynamics}
\begin{align}
    J_{\delta_i}\Delta\dot{\delta_i} &= -D_{\delta_i}\Delta \delta_i-\Delta P_i +K_i\Delta\delta_i\label{eq:droop_dynamics_a},\\
    J_{E_i}\Delta\dot{E_i}&= -D_{E_i}\Delta E_i-\Delta Q_i\label{eq:droop_dynamics_b}, \\
    &\eqref{eq:P_def}\text{ and }\eqref{eq:Q_def}\text{ hold true,}
\end{align}
\end{subequations}
where $\Delta E_i =E_i-E_i^*$;  $\Delta \delta_i =\delta_i-\delta_i^*$; $\Delta P_i =P_i-P_i^*$,  $\Delta Q_i =Q_i-Q_i^*$; $E_i^*$, $\delta_i^*$, $P_i^*$, $Q_i^*$ are the nominal set point values of the voltage magnitude, phase angle, active power and reactive power at the $i$\text{th} microgrid, respectively; $J_{\delta_i}$ and $J_{E_i}$ are tracking time constants; $D_{\delta_i}$ and $ D_{E_i}$ are droop coefficients; and $K_i$ is an output-feedback gain matrix. It is easy to verify that the above dynamics conform to the general form of dynamics given in \eqref{eq:subsystem}-\eqref{eq:linear-couple_eq}  with $x = [\Delta \delta_1, \Delta E_1, \ldots, \Delta \delta_n, \Delta E_n]^{\top}$, $h(x, u, v) = x$, $u= [\Delta P_1, \Delta Q_1, \ldots, \Delta P_n, \Delta Q_n]^{\top}$, $v = [K_1\Delta\delta_1, \ldots, K_n\Delta\delta_n]$. The functions  ${f}$  can derived from \eqref{eq:droop_dynamics_a}-\eqref{eq:droop_dynamics_b} and $g$  can derived from \eqref{eq:P_def}-\eqref{eq:Q_def}. 

\subsubsection{Quadratic  droop control}
\label{sec:quad_droop}

{Quadratic droop control} \cite{simpson2016voltage} offers  improvement both in theory and practice compared to the conventional droop control. The corresponding  dynamics at $i$th microgrid is given by
\begin{align}
\label{eq:quad_droop_dynamics}
    J_{E_i}\Delta\dot{E_i}&= -D_{E_i}\Delta E_i\big(\Delta E_i+E_i^*\big)-\Delta Q_i + K_i\Delta E_i,
\end{align}
such that  \eqref{eq:Q_def} holds true. Similar to the previous case, this dynamics also conform to  general form given  in \eqref{eq:subsystem}-\eqref{eq:linear-couple_eq}, with  $x = [\Delta E_1, \ldots, \Delta E_n]^{\top}$, $h(x,u) = x$, $u= [\Delta Q_1, \ldots, \Delta Q_n]^{\top}$, $v = [K_1\Delta E_1, \ldots, K_n\Delta E_n]$.   The functions  ${f}$  can derived from \eqref{eq:droop_dynamics_a}-\eqref{eq:droop_dynamics_b} and $g$  can derived from \eqref{eq:P_def}-\eqref{eq:Q_def}. 


In the remainder of this section, we consider specific test cases and validate the effectiveness of our algorithms for both conventional and quadratic droop control cases.

\subsection{Test Case I: Three microgrids network with conventional  angle droop control}

As standard in the power system literature, we assume that the dynamics of voltage magnitudes are predominantly slow  while the phase angles have fast dynamics, which is formalized by assuming that  $J_{E_i} \gg J_{\delta_i}$ in  \eqref{eq:droop_dynamics}. This is referred to as \textit{time scale separation} in power system literature \cite{winkelman1980multi}. With this assumption, we only focus on the behavior of the phase angle $\Delta \delta_i$ and active power $P_i$ for each subsystem $i \in\{1,2,3\}$. The 
distribution line parameters values are selected as given in   \cite{huang2021neural}, and the control parameters and set-point values are given in Table \ref{table:3mg:control_and_dispatch_params}. For neural network training,  we use a two layer network,  training sample size  2000, and learning rate  0.01. The  parameters for SMT based verification are $\epsilon_{l_i}=0.05$ and $\epsilon_{u_i}=0.8$ for $i \in \{1,2,3\}$. The tolerance level for the algorithm is selected as $\eta = 10^{-6}$. The algorithm converged  with only three iterations of the outer loop.



\begin{table}[b]
\caption{Control parameters and reference setpoints of the three microgrids network   \cite{huang2021neural}}
\label{table:3mg:control_and_dispatch_params}
\vspace{-1em}
\begin{center}
\begin{tabular}{ c|c|c|c } 
 \hline
   & MG 1 & MG 2 & MG 3 \\
 \hline
 $J_{\delta_i}$ &1.2 & 1.0 & 0.8\\ 
 \hline
 $D_{\delta_i}$ &1.2   & 1.2 & 1.2\\
 \hline
 $\delta_i^*$ (deg.)& 0& $55.67$ &$-45.37$\\
 \hline
 $E_i$ (p.u.) & 1& 1.05& 0.95 \\
 \hline
  $P_i^*$ (p.u.)& 0.1706 & 1.4578 &-0.0013\\
 \hline
\end{tabular}
\end{center}
\vspace{-1em}
\end{table}

\begin{figure} 
    \centering
  \subfloat[]{%
       \includegraphics[width=0.48\linewidth]{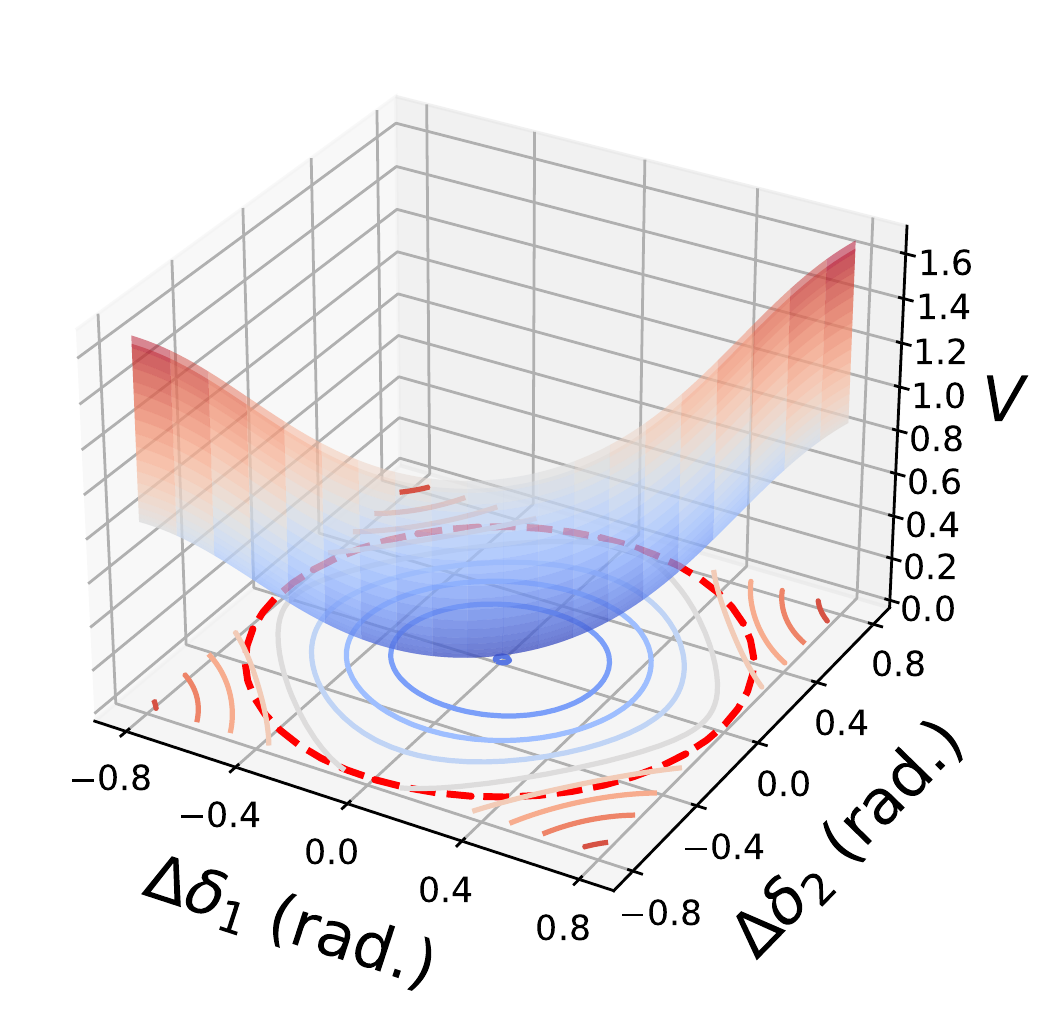}}
    \hfill
  \subfloat[]{%
        \includegraphics[width=0.48\linewidth]{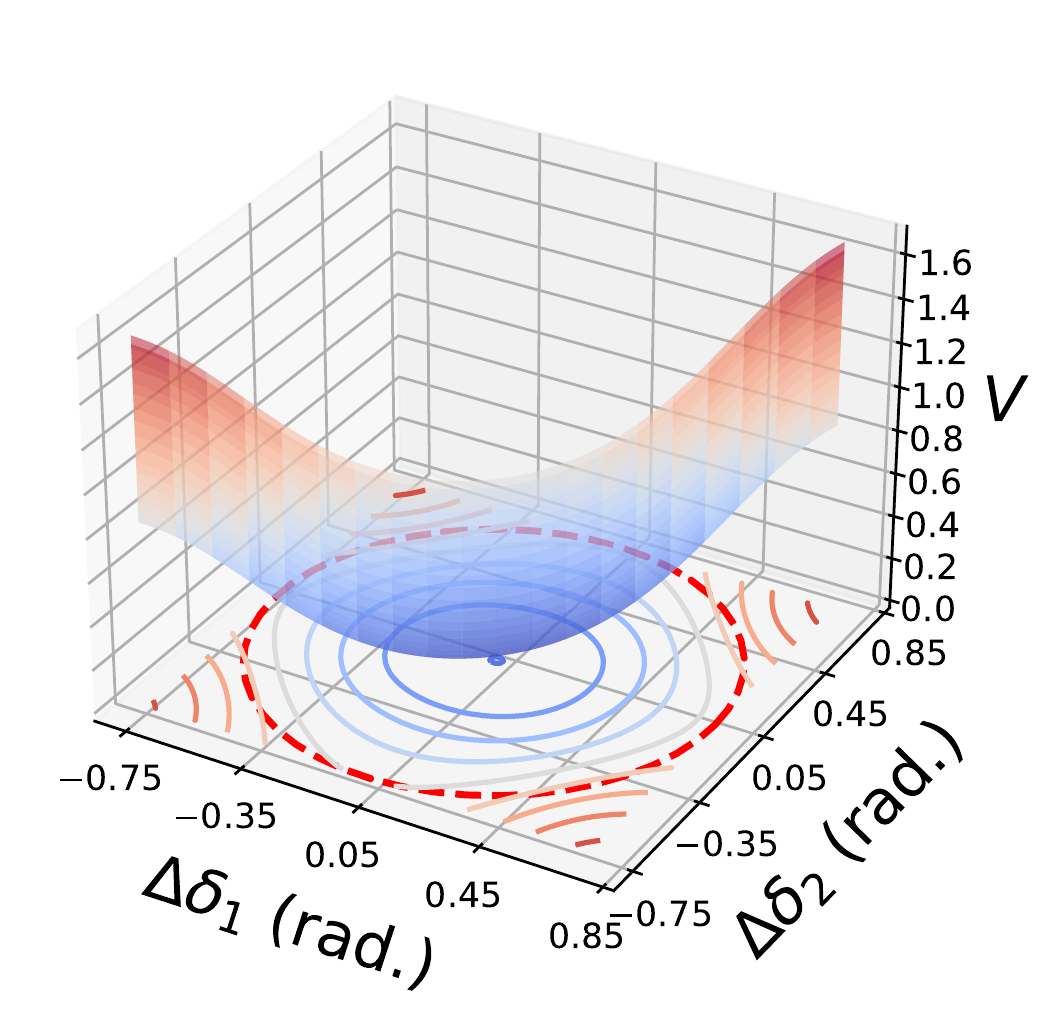}}
    \vfill
      \subfloat[]{%
       \includegraphics[width=0.48\linewidth]{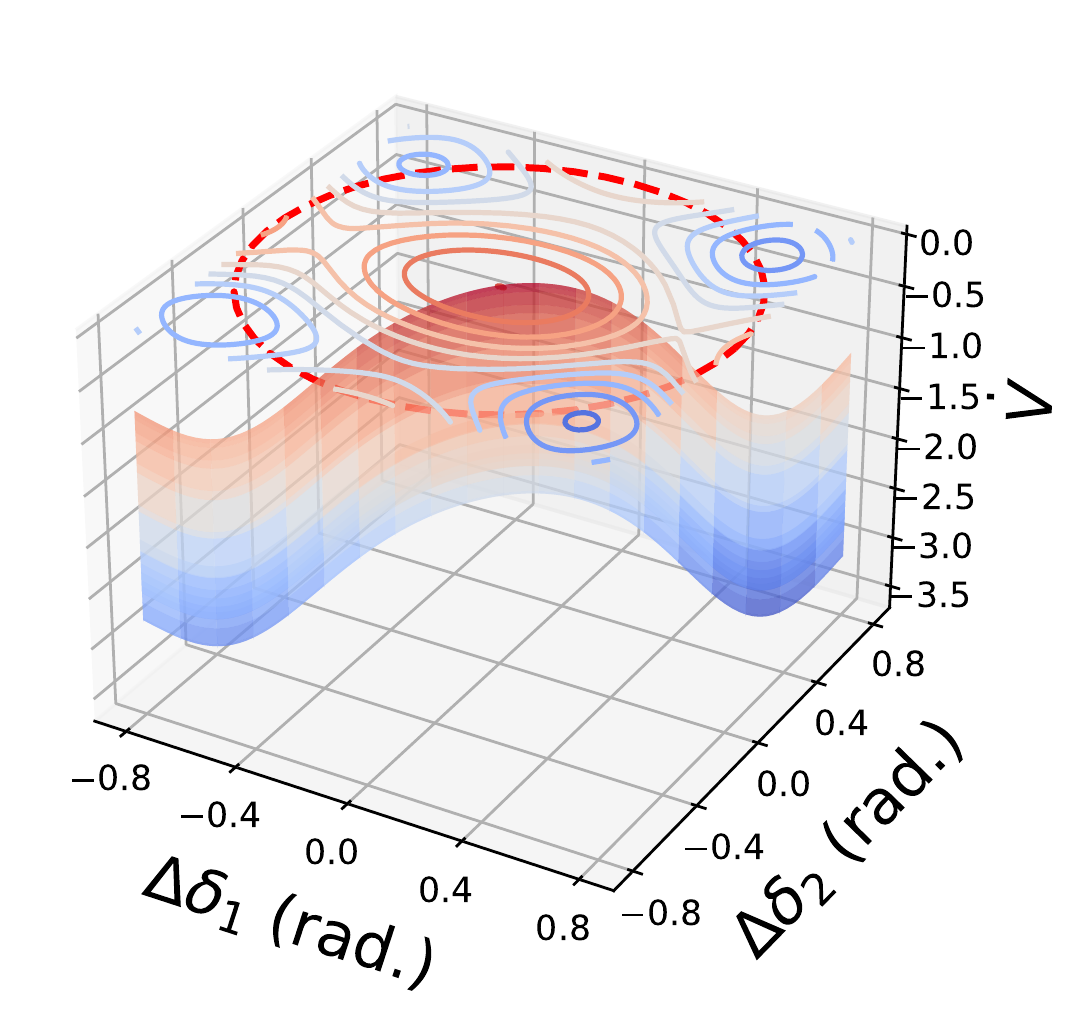}}
    \hfill
  \subfloat[]{%
        \includegraphics[width=0.48\linewidth]{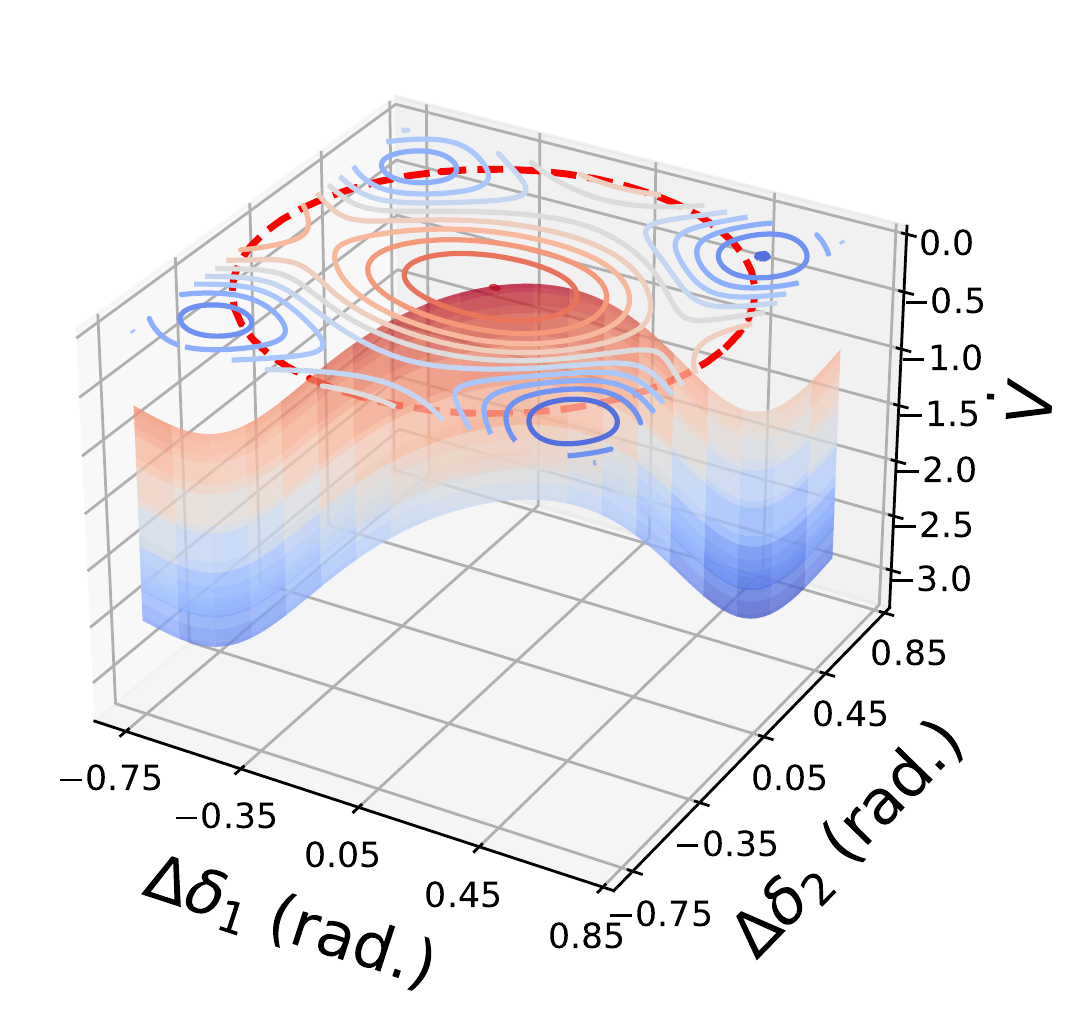}}
  \caption{Visualization of the distributed neural Lyapunov functions ($V$) and Lie derivatives ($\dot V$) for the three microgrids system. (a) Projection of DNCL function on $\Delta\delta_1$-$\Delta\delta_2$ plane; (b) Projection of DNL function on $\Delta\delta_1$-$\Delta\delta_2$ plane; (c) Projection of Lie derivative of DNCL function on $\Delta\delta_1$-$\Delta\delta_2$ plane; (d) Projection of Lie derivative of DNL function on $\Delta\delta_1$-$\Delta\delta_2$ plane. The  functions are positive  and their Lie derivatives are negative, resulting in a valid local Lyapunov functions.}
  \label{fig:3mg:lyap_3d_plots} 
\end{figure}

Fig.~\ref{fig:3mg:lyap_3d_plots}(a) shows the DNCL function obtained using our algorithm, projected on the $\Delta \delta_1$-$\Delta \delta_2$ plane.  Similarly, Fig.~\ref{fig:3mg:lyap_3d_plots}(b) shows DNL function obtained using our algorithm. Fig.~\ref{fig:3mg:lyap_3d_plots}(c) and Fig.~\ref{fig:3mg:lyap_3d_plots}(d) shows the Lie derivative ($\dot{V}$) of the  DNCL function and  DNL function, respectively. We note that to compute Lie derivatives,  we indeed considered the original nonlinear dynamics instead of the linear approximation for coupling. It is clear from these figures that the both DNCL and DNL functions  are positive and their Lie derivatives are negative over the specified domain, and hence they  satisfy the requirements to  be a valid Lyapunov function as given in Definition \ref{def:lyapunov}. The plots of these Lyapunov functions and their Lie derivatives on rest of the projection planes follow the same pattern as in Fig. \ref{fig:3mg:lyap_3d_plots}, and are omitted for brevity.

Fig.~\ref{fig:3mg:ROA_comp} shows the RoAs estimated using the Lyapunov functions obtained using our algorithm and other approaches mentioned in the beginning of this section.  It is clear from the figue that, $(i)$ the  ROAs of DNCL and DNL are significantly larger than the one obtained through the classical approaches  QCL and QL, and $(ii)$ the RoA of DNCL is slightly larger than that of the DNL because the additional secondary controller in DNCL. We note that the effect of secondary controller on getting larger RoA will become more prominent as the system dimension increases, as we show in the next case studies.  Thus, we show that our approach  is less conservative in stability assessment as compared to traditional methods. Not surprisingly, the ROAs of  DNCL and DNL are smaller than that of the centralized methods CNCL and CNL because of the performance limitations of distributed methods. However our distributed method offers significant reduction in the training time especially for large system,  as shown in  Table \ref{table:compute_time_comp}.

\begin{figure}
    \centering
    \includegraphics[width =\linewidth]{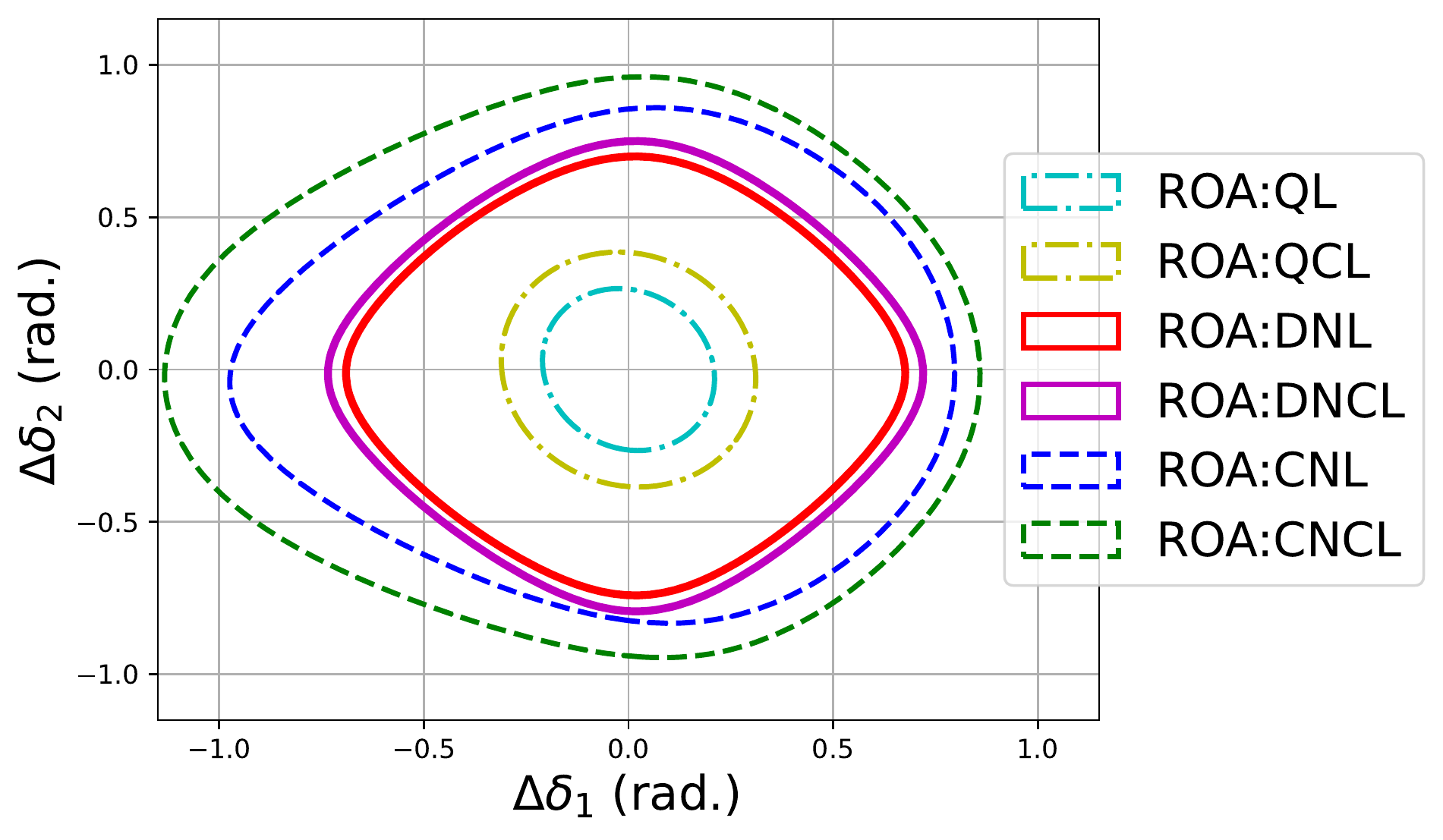}
    \caption{Comparison of the estimated region of attraction (RoA) from different approaches for the three microgrid network. Projections on $\Delta\delta_1$-$\Delta\delta_2$ space (with $\Delta\delta_3 =0$) is shown here.}
    \label{fig:3mg:ROA_comp}
\end{figure}



\subsection{Test Case II: Four microgrids network with quadratic voltage droop  control}


\begin{table}
\caption{Control parameters and reference setpoints of the four microgrids network \cite{teixeira2015voltage}}
\label{table:4mg:control_and_dispatch_params}
\vspace{-1em}
\begin{center}
\begin{tabular}{ c|c|c|c|c } 
 \hline
   & MG 1 & MG 2 & MG 3 & MG 4\\
 \hline
 $J_{\delta_i}$ &1 &1  & 1 & 1\\ 
 \hline
 $D_{\delta_i}$ &0.2 & 0.2 & 0.2 & 0.2\\
 \hline
 $\delta_i$ (deg.) &$0$    & $-6.30$ & $-3.72$ &$-10.027$\\
 \hline
 $E_i^*$ (p.u.) & 1& 1& 1 &1 \\
 \hline
  $Q_i^*$ (p.u.)& 0.013 & -0.013 & 0.017 & -0.009\\
 \hline
\end{tabular}
\end{center}
\vspace{-1em}
\end{table}

We next consider a line network of four microgrids with quadratic voltage droop control \eqref{eq:quad_droop_dynamics}. All the set-point values and control parameters are shown in table \ref{table:4mg:control_and_dispatch_params} and the distribution line parameters are taken from \cite{teixeira2015voltage}.   For neural network training,  we use a two layer network,  training sample size  2000, and learning rate  0.01. The  parameters for SMT based verification are $\epsilon_{l_i}=0.06$ and $\epsilon_{u_i}=0.7$. The tolerance level for the algorithm is selected as $\eta = 10^{-6}$. The algorithm converged  with only four iterations of the outer loop.

\begin{figure}[h]
    \centering
    \includegraphics[width =\linewidth]{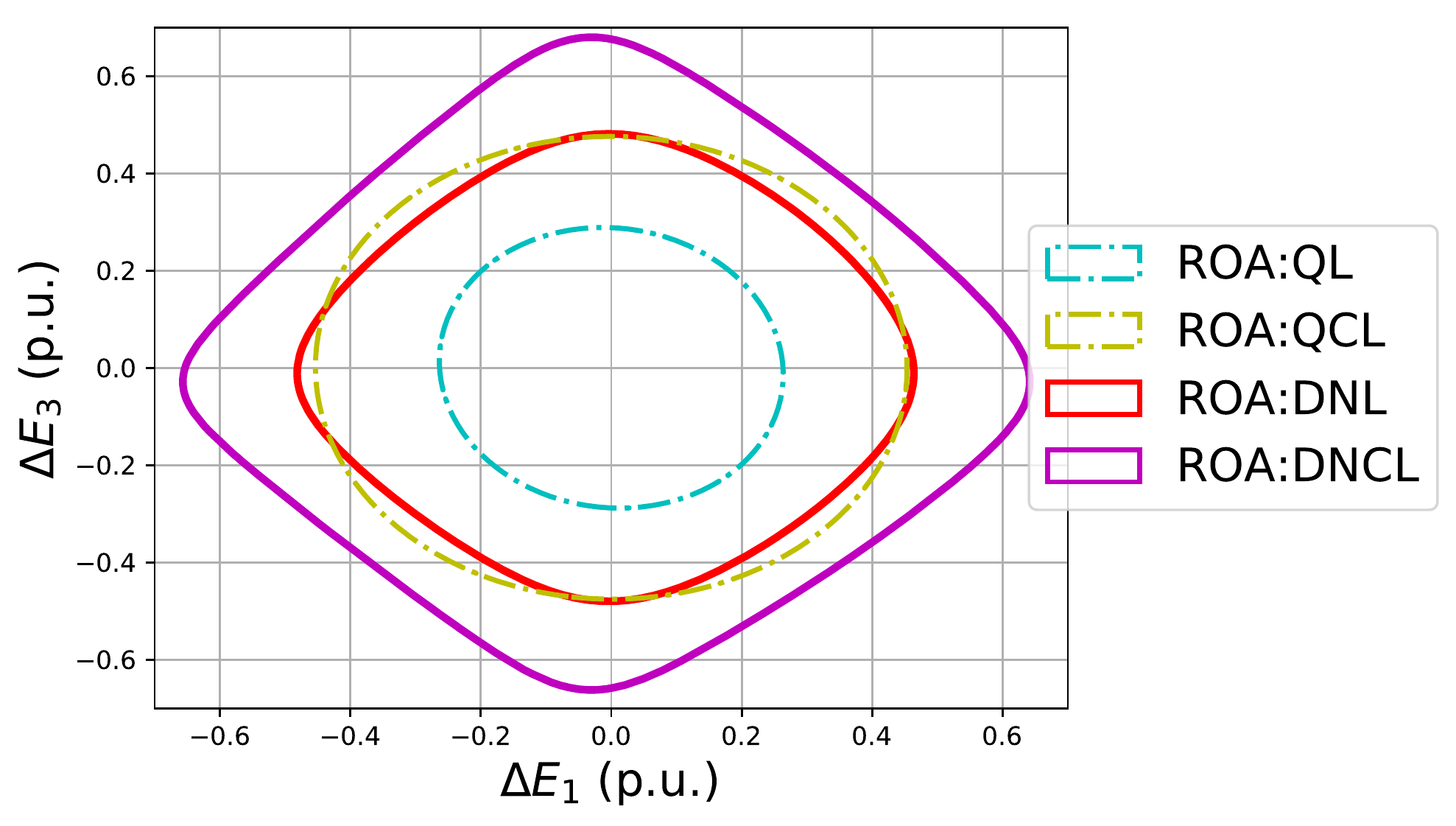}
    \caption{Comparison of the estimated region of attraction (RoA) from different approaches for the four microgrids network. Projections  on $\Delta E_1$-$\Delta E_3$ plane (with $\Delta E_2 =\Delta E_4 =0$) are shown.}
    \label{fig:4mg:ROA_comp}
\end{figure}


Fig.~\ref{fig:4mg:ROA_comp} shows the RoAs estimated using the Lyapunov functions obtained using our algorithm and other approaches mentioned in the beginning of this section. We can make the following observations: $(i)$ DNCL has larger ROA  compared to QCL and QL, $(ii)$ DNL's  ROA is larger than QL but  is  comparable with QCL,  and $(iii)$ DNCL has larger ROA compared to DNL. ROAs of DNL and QCL are close to each other  due to an LQR based secondary controller associated with QCL. Thus, our methods are better than the  traditional methods for estimating the RoA. Because of the complexity in the dynamics due to quadratic droop control, the SMT verification associated with CNL and CNCL couldn't converge. Hence, we omit the RoAs of CNL and CNCL from Fig.~\ref{fig:4mg:ROA_comp}. Finally, we conclude that for such systems that have complex non-linear nodal dynamics, DNL and DNCL are better suitable than centralized neural Lyapunov approaches.

\subsection{Test Case III: IEEE 123-node test feeder network with conventional voltage and angle droop control}

\begin{figure}
    \centering
    \includegraphics[width = \linewidth]{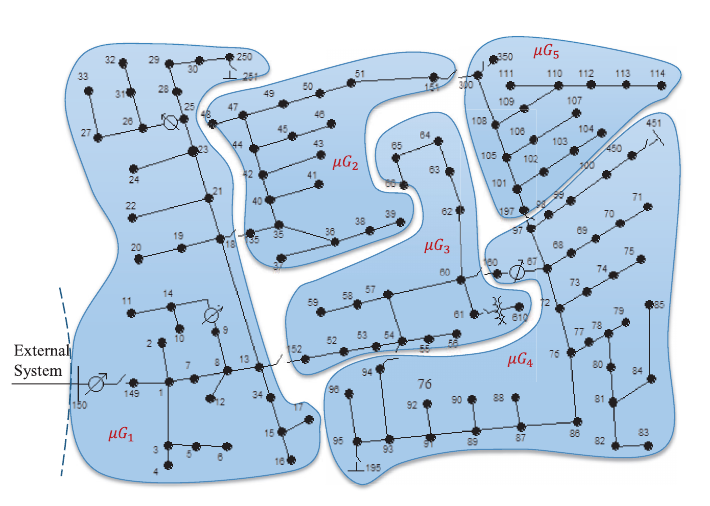}
    \caption{IEEE 123-node test feeder converted into a five-microgrid system \cite{sivaranjani2020distributed}}
    \label{fig:fiveMG_system}
\end{figure}

\begin{table}[b]
\caption{Control parameters and reference set-points of IEEE 123-node test feeder network \cite{sivaranjani2020distributed}}
\label{table:5mg:control_and_dispatch_params}
\vspace{-1em}
\begin{center}
\begin{tabular}{ c|c|c|c|c|c } 
 \hline
   & MG 1 & MG 2 & MG 3 & MG 4 & MG 5\\
 \hline
 $J_{\delta_i}$ &8 &12  & 12 & 9 & 10\\ 
 \hline
 $J_{E_i}$ &12.9 &10.2  & 11.56 & 10.83 & 11.73\\ 
 \hline
 $D_{\delta_i}$ &2.356 & 2.2 & 2.356 & 2.356 & 2.08\\
 \hline
  $D_{E_i}$ &2.50 & 2.0 & 2.222 & 2.083 & 2.272\\
 \hline
 $\delta_i^*$ (deg.) &$0$    & $0.233$ & $0.110$ &$0.158$ &$0.052$\\
 \hline
 $E_i^*$ (p.u.) & 1& 1.003& 1 &1.003 &0.999\\
 \hline
  $P_i^*$ (p.u.)& -0.13 & 0.57 & -0.25 & 0.53 & -0.72\\
 \hline
  $Q_i^*$ (p.u.)& 0.88 &-0.01 & -0.1 & 0.08 & -0.85\\
 \hline
\end{tabular}
\end{center}
\vspace{-1em}
\end{table}

Fig.~ \ref{fig:fiveMG_system} shows the IEEE 123-node test feeder system partitioned into five microgrids. We do not assume any time scale separation for this case, and consider conventional voltage and angle droop controls as primary control schemes. The set-point values and control parameters are given in Table \ref{table:5mg:control_and_dispatch_params}, and the distribution line parameters are selected as given in  \cite{sivaranjani2020distributed}. The system dynamics is as specified in  \eqref{eq:droop_dynamics}. We linearize the input-output coupling relationship around origin $o$ and convert the dynamics to the standard form as given in \eqref{eq:subsystem}-\eqref{eq:linear-couple_eq}. For neural network training,  we use a two layer network,  training sample size  4000, and learning rate  0.01. The  parameters for SMT based verification are $\epsilon_{l_i}=0.06$ and $\epsilon_{u_i}=0.6$. The tolerance level for the algorithm is selected as $\eta = 10^{-6}$. The algorithm converged  with only four iterations of the outer loop.   



\begin{figure}
    \centering
    \includegraphics[width =\linewidth]{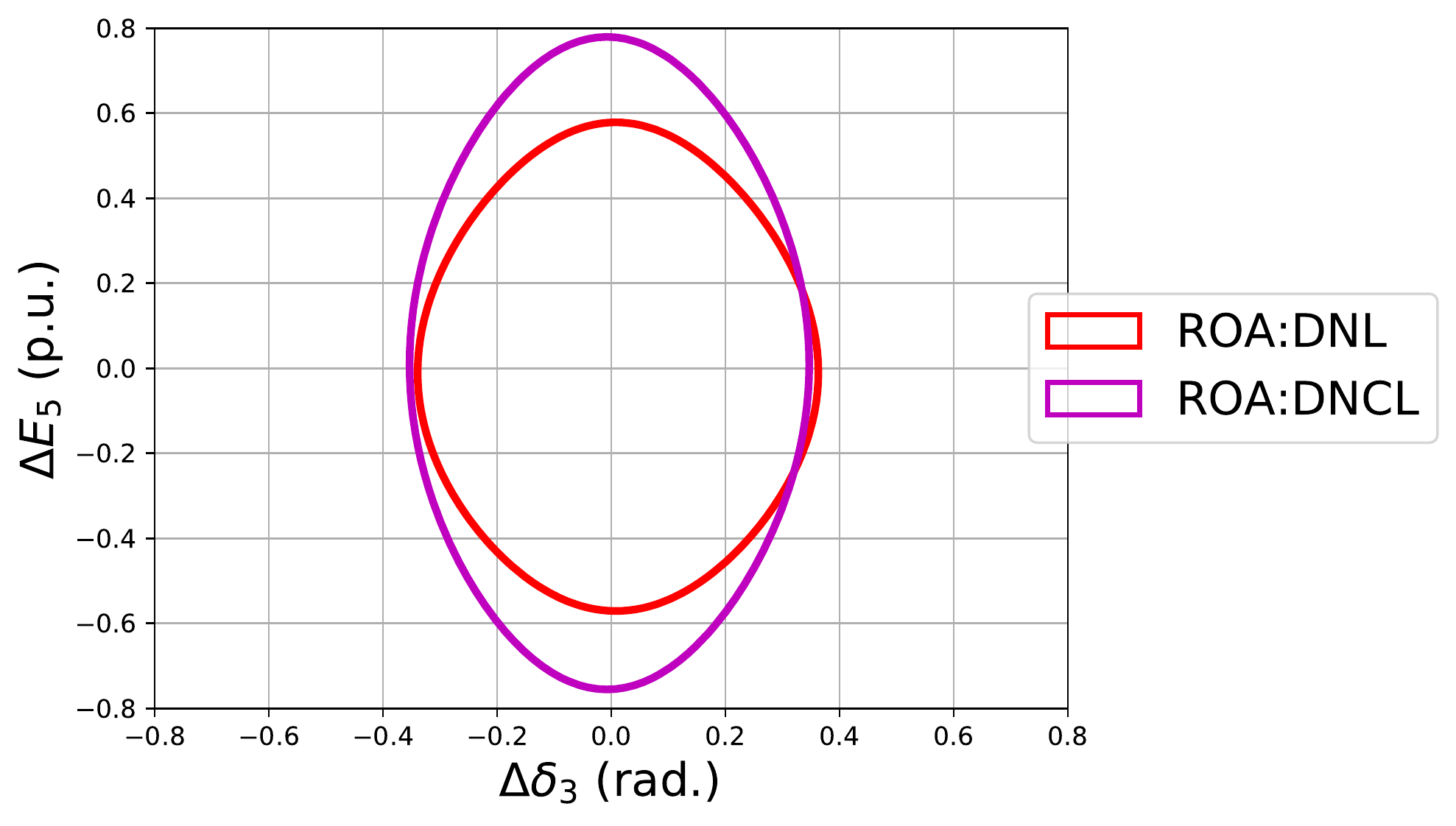}
       \caption{Comparison of the estimated region of attraction (RoA) from DNCL and DNL for the IEEE 123-node test feeder network. Projections on $\Delta\delta_3$-$\Delta E_5$ plane (with $\Delta\delta_1$ = $\Delta\delta_2$ = $\Delta\delta_4$ = $\Delta\delta_5$ = $\Delta E_1$ = $\Delta E_2$ = $\Delta E_3$ = $\Delta E_4$ = 0) are shown.}
    \label{fig:5mg:ROA_comp}
\end{figure}

Fig.~\ref{fig:5mg:ROA_comp} shows the RoAs estimated using our DNCL and DNL functions. Similar to the previous case studies, the RoA for DNCL is larger than that of DNL due to the secondary controller. We note that due to the large scale nature of the system, the RoAs estimated using the classical  QCL and QL approaches are significantly smaller (and barely visible) compared to that of DNCL and DNL, and hence we have omitted them from the figure. Moreover, again due to the large scale nature of the system, the SMT-based verification for the CNCL and CNL failed to converge even after a long time. Since  CNCL and CNL functions learning did not converge, we were also not able to estimate any RoAs based on these functions, and they are omitted from the figure. 
We note that this further emphasizes the computational benefit of our \textit{distributed}  approach, which converges fast even with SMT-based verification. Table \ref{table:compute_time_comp} shows the training time comparsion.



\begin{table}
\caption{Comparison of the training time for centralized (CNCL, CNL) and decentralized (DNCL, DNL) learning approaches.}
\label{table:compute_time_comp}
\vspace{-1em}
\begin{center}
\begin{tabular}{c|c|c|c|c} 
 \hline
   & \makecell{DNL\\(sec.)} & \makecell{DNCL\\(sec.)} & \makecell{CNL\\(sec.)}  & \makecell{CNCL\\(sec.)} \\
 \hline
 \makecell{IEEE \\123-node network} & 1679 & 1383.42 & $--$  & $--$ \\
   \hline
\end{tabular}
\end{center}
\vspace{-1em}
\end{table}

\subsection{Advantages Due to the Secondary Controller}

We observe that the learned secondary controller in the DNCL provides faster stabilization and robustness to perturbation, in addition to the larger RoA estimate, when compared to DNL where there is no secondary controller. The larger RoA of DNCL is clear from Fig.~\ref{fig:3mg:ROA_comp}, \ref{fig:4mg:ROA_comp}, and \ref{fig:5mg:ROA_comp}. The faster convergence is demonstrated through a time domain simulation in Fig.~\ref{fig:5mg:traj}. The increased robustness due to the secondary controller is demonstrated in Fig. \ref{fig:4mg:traj_noisy}. We add a small perturbation noise to the system state in the time domain simulation. From the figure, it is clear that the secondary controller is able to quickly suppress the disturbance.

\begin{figure}
    \centering
    \includegraphics[width =\linewidth]{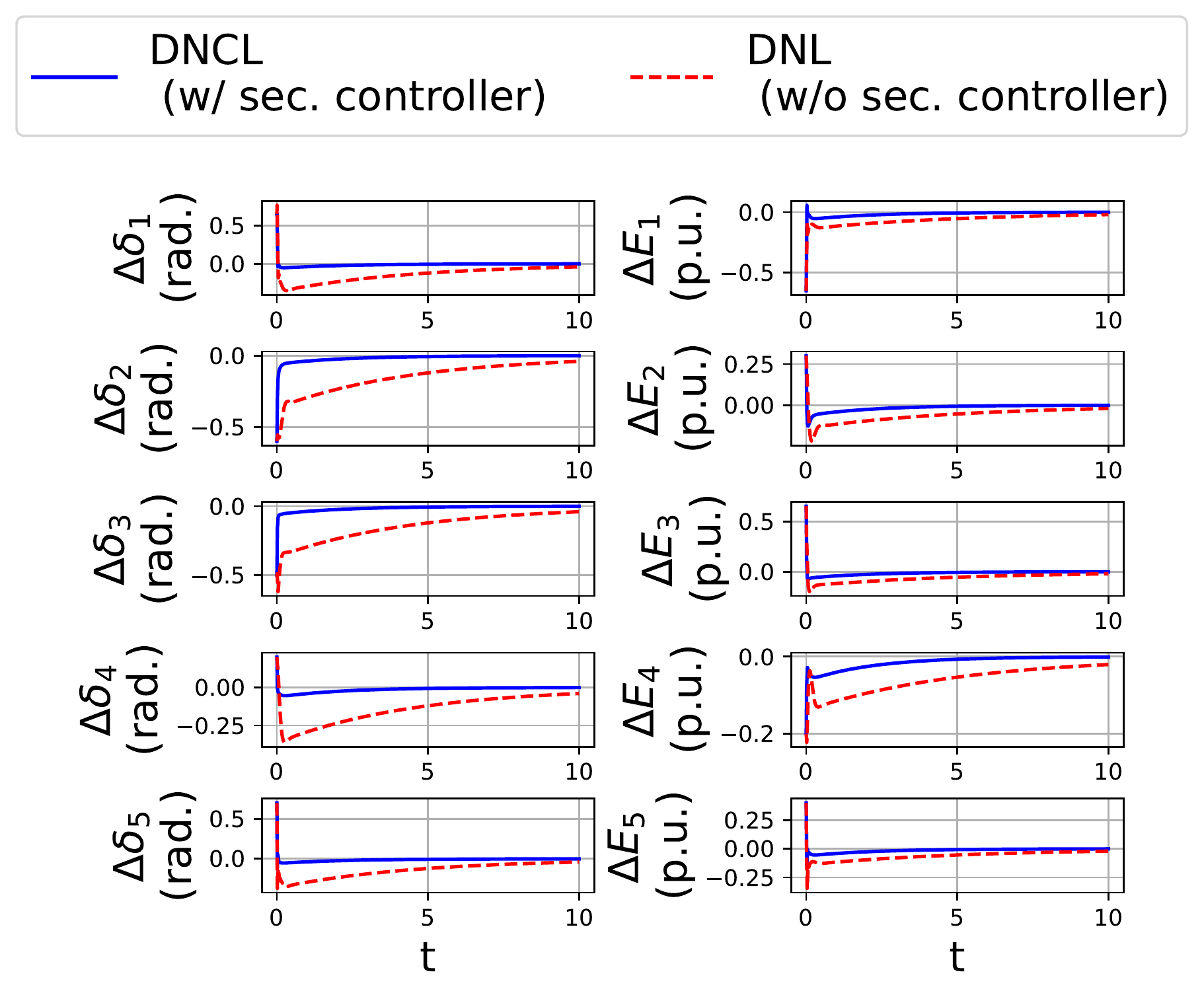}
    \caption{Time domain simulation comparison for the 123-node feeder  with initial state $x(0) = [0.65, -0.65, -0.6, 0.3, -0.5, 0.65, 0.2, -0.2, 0.7, 0.4]^\top$. The secondary controller (DNCL) stabilizes the system faster than DNL.}
    \label{fig:5mg:traj}
\end{figure}

\begin{figure}
    \centering
    \includegraphics[width =\linewidth]{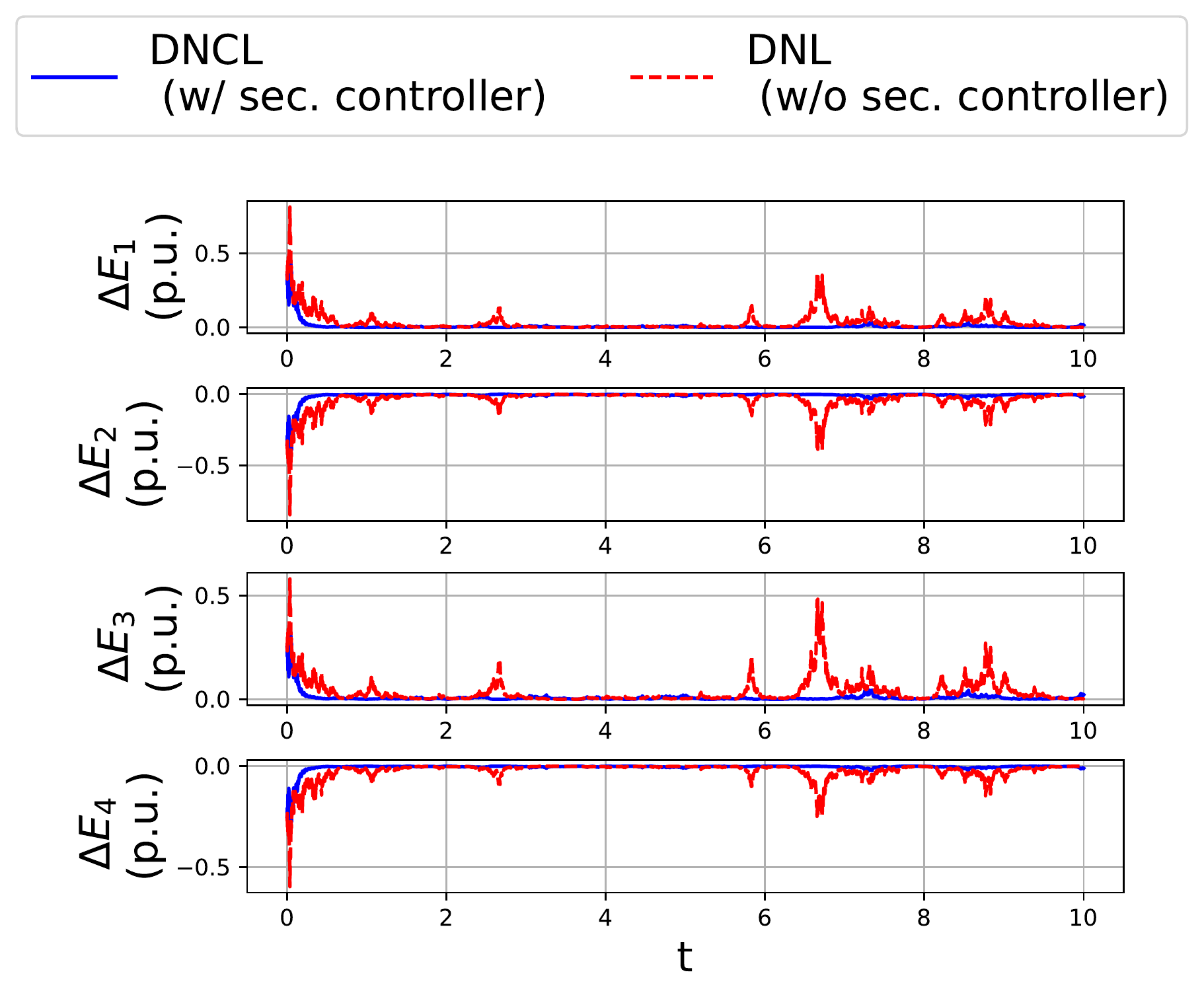}
    \caption{Time domain simulation comparison for four microgrid network  with  initial state $x(0) = [0.35, -0.35, 0.25, -0.25]^\top$. The secondary controller (DNCL) is able to suppress the disturbance resulting in a smoother trajectory.}
    \label{fig:4mg:traj_noisy}
\end{figure}

%% file: 06-Conclusion.tex
\section{Conclusion}
\label{sec:conclusion}

This paper proposes a novel distributed learning based framework for assessing Lyapunov stability of a class of networked nonlinear systems, where each subsystem is dissipative. The objective of the proposed framework is to construct a Lyapunov function in a distributed manner and to estimate the associated region of attraction for the networked system. We begin by leveraging a neural network function approximation to learn a storage function for each subsystem such that a local dissipativity property is satisfied by each subsystem. We next use a satisfiability modulo theories (SMT) solver based falsifier that verifies the local dissipativity of each subsystem by determining an absence of counterexamples that violate the local dissipativity property, as established by the neural network approximation. Finally, we verify network level stability by using an alternating direction method of multipliers (ADMM) approach to update the storage function of each subsystem in a distributed manner until a global stability condition for the network of dissipative subsystems is satisfied. We demonstrate the performance of the proposed algorithm and its advantages on three different case studies in microgrid networks. We note that this paper focuses on the synthesis of linear secondary controllers. Future work will focus on learning reinforcement learning based non-linear controllers that would further expand the RoA of the system and enhance its stability margins.